\documentclass{article}

\usepackage[preprint]{neurips_2021}

\usepackage[utf8]{inputenc} % allow utf-8 input
\usepackage[T1]{fontenc}    % use 8-bit T1 fonts
\usepackage{hyperref}       % hyperlinks
\usepackage{url}            % simple URL typesetting
\usepackage{booktabs}       % professional-quality tables
\usepackage{amsfonts}       % blackboard math symbols
\usepackage{nicefrac}       % compact symbols for 1/2, etc.
\usepackage{microtype}      % microtypography
\usepackage{xcolor}         % colors

\usepackage{amsmath,amssymb,amsthm}

\newcommand{\indicator}[1]{\mathbf{1}_{\left[#1\right]}}

\newtheorem{theorem}{Theorem}
\newtheorem{lemma}{Lemma}

\newtheorem{remark}{Remark}
\newtheorem{proposition}{Proposition}

\title{Bidding Agent Design in the LinkedIn Ad Marketplace}

\author{%
  Yuan Gao,  Kaiyu Yang,  Yuanlong Chen, Min Liu\\
LinkedIn Corporation\\
Mountain View, CA 94043\\
  \texttt{\{yugao,kyang,yuachen, mliu\}@linkedin.com} \\
  \And
  Noureddine El Karoui\thanks{Work done in part while at LinkedIn, Inc. and UC, Berkeley.}\\
  Berkeley, CA 94720\\
  \texttt{nkarouiprof@gmail.com} \\
}

\begin{document}

\maketitle

\begin{abstract}
We establish a general optimization framework for the design of automated bidding agent in dynamic online marketplaces.  It optimizes solely for the buyer's interest and is agnostic to the auction mechanism imposed by the seller. As a result, the framework allows, for instance, the joint optimization of a group of ads across multiple platforms each running its own auction format. Bidding strategy derived from this framework automatically guarantees the optimality of budget allocation across ad units and platforms. Common constraints such as budget delivery schedule, return on investments and guaranteed results, directly translates to additional parameters in the bidding formula. We share practical learnings of the deployed bidding system in the LinkedIn ad marketplace based on this framework. 
\end{abstract}

\section{Introduction}
There is a growing trend of automation in online advertising.  The manual task of creating tailored ad for each platform and splitting budget across these platforms are largely automated by programmatic solutions in demand-side platforms (DSPs). The LinkedIn Marketing Solutions is a DSP that allows advertisers to reach LinkedIn members through the LinkedIn Feed as well as its audience network consists of trusted third party publishers.  Impression opportunities on these various placements are normally sold through different sorts of auctions. For instance, the LinkedIn Feed placement employs generalized second price (GSP) auctions with reserve price, while majority of publishers in the audience network now uses first price auctions. The paper deals with the design of an automated agent that places bid on behalf of advertisers in these different types of auctions. The goal is to help advertisers maximize their results under natural constraints such as budget delivery plan and expected return on investment (ROI). 

Typical advertising objectives include impression, reach, click, social engagement, video view, lead and conversion.  Advertisers can directly encode their valuations of the objectives into the bids for each impression opportunity. For example,  advertisers who value \$5 per click for their ads could place a bid of \$5 * p(click) in a second price auction, where p(click) is the probability that the user would click after seeing the ad. However this bidding strategy is blind to the budget. This could lead to early termination if the budget is constrained. Throttling-based systems are introduced \cite{AGWY2014, X2015} to address this problem by probabilistically blocking the ad from entering into auctions as a way to smooth out the spend over the entire lifetime. Yet a better strategy in this case is to lower the bids so that budget is spent more effectively.

To formally model this problem, in economics literature \cite{BBW2015, BG2019} a profit-maximization setup is often adopted, where profit is defined as the difference between valuation and payment.  In practice, however, valuations from advertisers are not always available. Defining valuations itself is a difficult problem given the scale and complexity of the targeted audience. A preferred way is to frame the bidding problem in a result maximization setup under budget constraint. Valuations from advertisers, if available, can be encoded in this setup via additional constraints on the bids or ROIs. 

Mathematically such a bidding problem with budget constraint can be cast in a stochastic optimal control framework \cite{AKKS2012, GKP2013, BBW2015, FGL2017}. When budget is relatively large compared to the magnitude of bids, fluid approximation \cite{GVR1994, BBW2015, GKP2013, FGL2017} is usually applied which greatly simplifies the solution. Explicit bidding formulas are given in \cite{Z2014} under further assumptions on stationary market competition. In practice though, competition in online ad marketplace is highly dynamic due to the change in demand and supply. For second price auctions, dynamic bid update strategies without stationary assumptions are studied in \cite{BG2019, K2020}.

Instead of directly optimizing for individual advertiser's interest, there is rich literature dealing with the optimization towards objectives of the platform or DSP, by means of online resource allocations \cite{M2007, ABM2019, CBAD2011, AKMT2008, B2020, G2017, F2010}. Bidding strategy for each advertiser in turn arises as a by-product of this optimization. These centralized mechanisms and their equilibrium studies are not considered in our design.

There is limited study on the optimal budget allocation problem across multiple placements. In the case when all placements run second price auctions,  a stochastic bandit based solution is proposed \cite{FB2021} where the value for each impression opportunity is assumed to be unknown. This is an overly pessimistic assumption since in reality the value is usually given by a response prediction model. A closely related problem is studied in \cite{Z2016}, in which the solution relies on feedback controllers towards pre-calculated ROI target for each placement.

In Section \ref{general-formulation} we lay out bidding agent design under a general framework of result maximization under budget constraint.  Online methods for solving this optimization problem are discussed in Section \ref{online-methods-in-the-dual}. Section \ref{other-types-of-constraint} demonstrates that features such as cost control, budget delivery plan and guaranteed delivery easily translates to constraints in this optimization framework, so that they can be incorporated holistically in one bidding formula. 

It's worth noting that the framework is inherently compatible with a collection of ad units and a mixture of different types of selling mechanisms, thereby allowing for the joint optimization of a group of ads across multiple placements each running (possibly) a different auction format. In particular, in Section \ref{marginal_roi} we show that optimal bidding solution derived from this formulation automatically guarantees the optimality of budget allocation across ad units and placements. 

Finally, in Section \ref{practical-considerations} explicit formulas for the starting bid is derived and various issues around the implementation of online optimization methods are discussed.
\section{General Formulation}\label{general-formulation}
In this section we establish a general optimization framework for a single advertiser. For simplicity,  the problem is formulated as maximizing total result under a single budget constraint. As noted later in Section \ref{other-types-of-constraint} it is straightforward to incorporate other types of constraints as well. 

Suppose there are a total of $T$ impression opportunities, triggered from targeted users' search or view sessions.  The $t$-th impression has a value of $v_t$ for the agent, which is usually produced from a user response prediction model tied to the advertiser's marketing objective. As an example, for an advertiser optimizing for conversions, $v_t$ represents the impression to conversion rate.  Denote $G_t(b)$ as the probability of winning the $t$-th impression under a bid price of $b$ and $H_t(b)$ as the expected cost.  With an advertising budget $B$ the optimization problem becomes
\begin{equation}\label{eq:general-formulation}
\max_{b_t} \sum_{t=1}^T v_tG_t(b_t),\quad s.t.\ \sum_{t=1}^T H_t(b_t) \le B.
\end{equation}
In a complete information setting, we can model  $G_t=\indicator{b_t \ge c_t}$ and $H_t=d_t\indicator{b_t \ge c_t}$, where $c_t$ is the minimum bid to win impression $t$, and $d_t$ is the cost if won.  This formulation leads to a Knapsack problem \cite{C2008, ZN2008}. In practice it's often convenient to assume distributional foresight, where $G_t$ and $H_t$ are modeled as smooth functions. Denote their derivatives as $g_t$ and $h_t$. We make the following assumptions:
\begin{itemize}
\item
  $G_t(0) = 0$, $G_t(+\infty) = 1$ and $g_t > 0$ (monotonicity).
\item
  $H_t(0) = 0$ and $h_t \ge 0$.
\end{itemize}
These are mild assumptions that generally hold for all practical auction types including first \& second price auctions, generalized first price (GFP) auctions , generalized second price (GSP) auctions,  Vickrey–Clarke–Groves (VCG) auctions and variants of them with floor price.

The Lagrangian for \eqref{eq:general-formulation} is
\begin{equation}\label{eq:lagrangian}
L(b_t, \lambda) = \sum_{t=1}^T\left[v_tG_t(b_t) - \lambda H_t(b_t)\right] + \lambda B.
\end{equation}
At fixed $\lambda > 0$, we would choose $b_t$ that maximizes the surplus
\begin{equation}\label{eq:surplus_maximization}
(v_t/\lambda)G_t(b_t) - H_t(b_t),
\end{equation}
where $v_t/\lambda$ is the $\lambda$-adjusted value. Optimality condition in \(b_t\) implies
\begin{equation}\label{eq:optimality-b}
v_tg_t(b^*_t) = \lambda h_t(b^*_t),
\end{equation}
or equivalently,
\begin{equation}\label{eq:b-star}
b_t^*(\lambda) = \left(\frac{h_t}{g_t}\right)^{-1}\left(\frac{v_t}{\lambda}\right),
\end{equation}
where $(\cdot)^{-1}$ denotes function inverse.
\begin{remark}[Explicit formula for first and second price auctions]\label{rm:explicit_formula_for_1st_and_2nd}
In ad marketplace impression opportunities are usually sold via auctions, and $H_t$ and $G_t$ are generally related in various kinds of auctions. In second price auctions, for example, $H_t(b_t) = \int_0^{b_t} b g_t(b)db = b_tG_t(b_t) - \int_0^{b_t}G_t(b)db$. This implies $h_t(b_t) = b_tg_t(b_t)$. According to \eqref{eq:optimality-b} The optimal bidding strategy in second price auctions is simply bidding the $\lambda$-adjusted value
\[
b_t^*(\lambda) = v_t/\lambda.
\]
In first price auctions $H_t(b_t) = b_tG_t(b_t)$. The surplus maximization \eqref{eq:surplus_maximization} simplifies to $\max_{b_t} (v_t/\lambda - b_t)G_t(b_t)$.  Since $h_t(b_t) = G_t(b_t) + b_tg_t(b_t)$, \eqref{eq:b-star} gives the following optimal bidding strategy
\[
b_t^*(\lambda) = \left(I + \frac{G_t}{g_t}\right)^{-1}\left(\frac{v_t}{\lambda}\right),
\]
where $I$ stands for identity mapping. Note that $G_t/g_t \ge 0$, so we always have $b_t^*(\lambda) \le v_t/\lambda$. This means that in first price auctions the bid has to be shaded on top of the $\lambda$-adjusted value.
\end{remark}
We define two useful quantities that will be used throughout.
\[
V_t(\lambda) := v_tG_t(b_t^*(\lambda)),
\]
the expected value obtained in the $t$-th opportunity when bidding optimally under $\lambda$, and
\[
S_t(\lambda) := H_t(b_t^*(\lambda)),
\]
the expected spend in the $t$-th opportunity when bidding optimally under $\lambda$.  Let $V(\lambda) := \sum_{t=1}^T V_t(\lambda)$ be the total expected value and $S(\lambda) := \sum_{t=1}^T S_t(\lambda)$ the total expected spend.

The following proposition shows that their derivatives are linearly related, and, under mild conditions, both of them are monotonic.
\begin{proposition}\label{prop:value_spend_relation}
$V_t'(\lambda)  \equiv \lambda S_t'(\lambda), \forall \lambda > 0$.  If in addition $\forall h_t > 0$, we have $(\log h_t)' > (\log g_t)'$, then $V_t'(\lambda)  \equiv \lambda S_t'(\lambda) \le 0,  \forall \lambda > 0$.
\end{proposition}
\begin{proof}
The relation $V_t'(\lambda)  \equiv \lambda S_t'(\lambda), \forall \lambda > 0$ is directly given by optimality condition \eqref{eq:optimality-b}.  Now apply implicit differentiation on \eqref{eq:optimality-b}, we get
\[
\frac{db_t^*}{d\lambda} = \frac{h_t(b_t^*)}{v_tg_t'(b_t^*) - \lambda h_t'(b_t^*)}.
\]
When $h_t(b_t^*) = 0$ it's clear that $\frac{db_t^*}{d\lambda} = 0$. Otherwise, the condition $(\log h_t)' > (\log g_t)'$ implies
\[
\frac{h_t'(b_t^*)}{h_t(b_t^*)} > \frac{g_t'(b_t^*)}{g_t(b_t^*)} = \frac{g_t'(b_t^*)}{(\lambda / v)h_t(b_t^*)} \Rightarrow \frac{db_t^*}{d\lambda} < 0.
\]
Finally by chain rule, we have $V_t'(\lambda) = v_tg_t(b_t^*)\frac{db_t^*}{d\lambda} \le 0$.
\end{proof}
\begin{remark}
In second price auctions, the extra condition $\forall h_t > 0, (\log h_t)' > (\log g_t)'$ in Proposition \ref{prop:value_spend_relation} always holds due to the fact that $h_t(b_t) = b_tg_t(b_t)$. In first price auctions, it's easy to show that the condition translates to log-concavity in $G_t$, i.e. $G_tG_t'' \le g_t^2$. Note that $G_t$ can be seen as the cumulative density function (CDF) of competitors' bid distribution, therefore the condition holds as long as the distribution is log-concave.
\end{remark}
Plugging $b_t^*(\lambda)$ in the Lagrangian \eqref{eq:lagrangian}, we arrive at the dual problem.
\begin{equation}\label{eq:dual_objective}
\min_{\lambda \ge 0}L(\lambda) := \sum_{t=1}^T\left[V_t(\lambda) - \lambda S_t(\lambda) + \lambda \frac{B}{T}\right].
\end{equation}
Setting $L'(\lambda^*) = 0$, and applying Proposition \ref{prop:value_spend_relation} we get
\begin{equation}\label{eq:dual_optimality}
\sum_{t=1}^T S_t(\lambda^*) = B.
\end{equation}
Equation \eqref{eq:dual_optimality} says that an optimal solution $\lambda^* > 0$ would result in a match of the expected spend $S(\lambda)$ with budget. In other words, the solution $\lambda^*$ is either 0 when budget is unconstrained, or one that spends the budget, not surprising given Karush--Kuhn--Tucker (KKT) conditions \cite{KKT2014}.

In addition, Proposition \ref{prop:value_spend_relation} implies the monotonicity of spend $S(\lambda)$. Due to this monotonicity, bisection method can be used to find $\lambda^*$ given $G_t$, $H_t$ and $v_t$ for $t = 1, \ldots, T$.  For example,  $\lambda^*$ hence the optimal bidding sequence $b_t$ can be obtained in hindsight given historical auction logs. However in practice it's difficult to generate accurate forecast for $G_t$ and $H_t$ in dynamic marketplaces. In the next section we consider online optimization methods in the dual that does not require access to explicit models of these quantities.

\section{Online Methods in the
Dual}\label{online-methods-in-the-dual}
In this section we present various online optimization methods in the dual \eqref{eq:dual_objective}. Denote the dual loss at impression opportunity $t$ as
\[
L_t(\lambda) := V_t(\lambda) - \lambda S_t(\lambda) + \lambda \frac{B}{T}.
\]
The dual objective \eqref{eq:dual_objective} is then $\sum_t L_t(\lambda_t)$. Derivative of the online loss 
\[
L_t'(\lambda) = B/T - S_t(\lambda),
\]
 is the difference between the average budget in each opportunity and the expected spend at $t$.
\subsection{Follow the Leader}\label{follow-the-leader}

In follow the leader (FTL) algorithm at each iteration we choose the
best solution in hindsight, i.e.,
\[
\lambda_{t+1} = \text{argmin}_{\lambda \ge 0}\sum_{\tau=1}^t L_\tau(\lambda).
\]
Similar to optimality condition \eqref{eq:dual_optimality} we have
\[
\frac{1}{t}\sum_{\tau=1}^t S_\tau(\lambda_{t+1}) = \frac{B}{T}.
\]
In practice, this can be achieved by replaying the past auctions to find $\lambda_{t+1}$ such that average cost per auction is $B/T$. As the problem size becomes larger over time, other variants such as those using a fixed size look back window from $t-\Delta t$ to $t$ might be preferred. 

In the stationary case (i.e. $G_t = G, H_t = H$), the FTL approach converges to the optimal solution in one iteration (i.e.
$\lambda_2 = ... = \lambda_T = \lambda^*$), but in reality $\lambda_t$ produced by FTL can be unstable when $G_t$ and $H_t$ are nonstationary.
\subsection{Linearization and Follow the Regularized
Leader}\label{linearization-and-follow-the-regularized-leader}
It's obvious that in both first and second price auctions $L_t(\lambda)$ is convex since $L_t''(\lambda) \ge 0$ according to Proposition \ref{prop:value_spend_relation}. Therefore the linearized loss around $\lambda_t$
\[
\tilde{L_t}(\lambda) := L_t(\lambda_t) + L'_t(\lambda_t)(\lambda - \lambda_t)
\]
is an lower bound of $L_t(\lambda)$. In follow the regularized leader
(FTRL) we choose
\[
\lambda_{t+1} = \text{argmin}_{\lambda \ge 0}\sum_{\tau=1}^t \tilde{L}_\tau(\lambda) + \Omega(\lambda),
\]
where $\Omega(\lambda)$ is a regularization term.
\begin{theorem}[Dual Online Mirror Descent]\label{thm:dogd}
When $\Omega(\lambda) = \lambda^2/(2\epsilon)$, the online update is additive $\lambda_{t+1} =  \lambda_t - \epsilon L'_t(\lambda_t)$; when $\Omega(\lambda) = \lambda\ln\lambda/\epsilon$, the online update is multiplicative $\lambda_{t+1} = \lambda_te^{-\epsilon\L'_t(\lambda_t)}$.
\end{theorem}
Second order approximations to $L(\lambda)$ can be employed to obtain faster convergence. In particular, $L''(\lambda) = -\sum_{t=1}^T S'_t(\lambda) = -S'(\lambda)$ represents the sensitivity of expected spend around $\lambda$.  Intuitively when the variation of expected spend is small one can make larger updates to $\lambda$. Online optimization methods that use second order information, such as natural gradient descent \cite{A2000} or online Newton's method \cite{HAK2007}, can be used in the update.
\section{Other Types of Constraints}\label{other-types-of-constraint}
Apart from the overall budget constraint, in reality there are sometimes preferences on budget delivery schedule as well.  Advertisers may also prefer some sort of cost control, e.g., keeping the cost per result under a certain threshold. Other types of constraints include the request to ensure certain number of results are delivered within a time range. In this section we show it's straightforward to incorporate them in the general formulation.
\subsection{Cost Control}
A cost per result target $C$ could be specified representing a given ROI goal. The agent attempts to control its cost per result under $C$ by solving the following problem
\[
\max_{b_t} \sum_{t=1}^T v_tG_t(b_t),\quad s.t.\ \sum_{t=1}^T H_t(b_t) \le \min\left(B, C\sum_{t=1}^T v_tG_t(b_t)\right).
\]
Equipping a multiplier for each constraint, the Lagrangian becomes
\[
\sum_{t=1}^T\left[(1+\mu C)v_tG_t(b_t)-(\lambda+\mu) H_t(b_t)\right] + \lambda B.
\]
Solution in the optimal bidding formula $b_t$ is similar, where
\[
b_t^*(\lambda, \mu) = \left(\frac{h_t}{g_t}\right)^{-1}\left(\frac{1 + \mu C}{\lambda + \mu}\cdot v_t\right).
\]
Again KKT conditions imply either $\lambda^* = 0$ (in which case the budget is not a binding constraint) or the budget is spent, i.e.,
\[
\sum_{t=1}^T H_t(b_t^*(\lambda^*, \mu^*)) = B.
\]
Similarly, either $\mu^* = 0$ (in which case the cost per result target is not a binding constraint), or the cost per result is equal to $C$, i.e.,
\[
\sum_{t=1}^T H_t(b_t^*(\lambda^*, \mu^*)) = C\sum_{t=1}^T v_tG_t(b_t^*(\lambda^*, \mu^*)).
\]
\subsection{Budget Delivery Control}
We could define budget delivery constraints on some subintervals $T_k, k = 1, \ldots, K$, where $\sum_{t\in T_k} \le T$. This can reflect advertisers' delivery preferences, for example, to limit the spend on weekends.
\[
\max_{b_t} \sum_{t=1}^T v_tG_t(b_t), \ s.t.\sum_{t=1}^T H_t(b_t) \le B, \sum_{t\in T_k} H_t(b_t) \le B_k, \forall k.
\]
In practice these constraints can also be dynamic, where they can be added and adjusted in realtime based on inputs from advertisers. The Lagrangian is then
\[
\sum_k\sum_{t\in T_k}\left[v_tG_t(b_t)-(\lambda+\lambda_k) H_t(b_t)\right] + \sum_k\lambda_kB_k + \lambda B,
\]
where $\lambda_k \ge 0$. The new constraints influence the bidding formula. For $t \in T_k$,
\[
b_t^*(\lambda, \lambda_k) = \left(\frac{h_t}{g_t}\right)^{-1}\left(\frac{v_t}{\lambda + \lambda_k}\right).
\]
Note that each $\lambda_k$ is only active during the period of $T_k$, i.e., when the constraint is in place.
\subsection{Guaranteed Delivery}
Advertisers may also request that certain number of results to be delivered, for instance, during holiday season. Instead of manually increasing the budget during the time period, more precise control can be achieved by encoding the requirement as constraints directly in the optimization framework. More formally, suppose we'd like $V_k$ number of results to be delivered in each subinterval $T_k$ (defined in the previous section), then the problem becomes
\[
\max_{b_t} \sum_{t=1}^T v_tG_t(b_t), \ s.t.\sum_{t=1}^T H_t(b_t) \le B, \sum_{t\in T_k} v_tG_t(b_t) \ge V_k, \forall k.
\]
The Lagrangian is now
\[
\sum_k\sum_{t\in T_k}\left[(1 + \mu_k)v_tG_t(b_t)-\lambda H_t(b_t)\right] - \sum_k\mu_kV_k + \lambda B.
\]
For $t \in T_k$, the bidding formula is
\[
b_t^*(\lambda, \mu_k) = \left(\frac{h_t}{g_t}\right)^{-1}\left(\frac{1 + \mu_k}{\lambda}\cdot v_t\right).
\]
Since $\mu_k \ge 0$, a nonzero multiplier $\mu_k^*$ would give a boost to the bid during the period $T_k$ to help achieve the delivery requirement.

A combination of the constraints can be implemented at the same time. As an illustration, if all constraints discussed in this section are involved, the bidding formula would then become
\[
b_t^*(\lambda, \mu, \lambda_k, \mu_k) = \left(\frac{h_t}{g_t}\right)^{-1}\left(\frac{1 + \mu C + \mu_k}{\lambda + \lambda_k + \mu}\cdot v_t\right).
\]
This offers a systematic way to handle a complex set of constraints simultaneously.
\section{Multiple Placements, Group of Ads and Equality of Marginal ROI}\label{marginal_roi}
In advertising, each impression opportunity is associated with a placement. For instance, an advertiser might want to serve his/her ad on multiple publishers at the same time, where each publisher website is a placement. Our framework naturally allows the simultaneous optimization on multiple placements as equation \eqref{eq:general-formulation} does not differentiate the $t$-th impression opportunity based on its placement. Concretely, suppose there are a total of $K$ placements and $T_k, k = 1, \ldots, K$ represents the set of impression opportunities associated with the $k$-th placement, with $\sum_{t \in T_k} = T$. Then \eqref{eq:general-formulation} is essentially
\[
\max_{b_t} \sum_{k=1}^K\sum_{t \in T_k} v_tG_t(b_t),\quad s.t.\ \sum_{k=1}^K\sum_{t \in T_k} H_t(b_t) \le B.
\]
Note that the impression opportunities across different placements can be interleaved in time.  According to the Lagrangian \eqref{eq:lagrangian} the total number of results obtained under optimal bidding is
\[
L^* := L(b^*_t, \lambda^*) = \sum_{k=1}^K\left[\lambda^* B_k + \sum_{t \in T_k} V_t(\lambda^*) - \lambda^* S_t(\lambda^*) \right],
\]
where $B_k$ is the budget spent on the $k$-th placement. In this case $\partial L^*/\partial B_k$, the marginal return on investment (ROI) for all placements $1, \ldots, K$ are equal (to $\lambda^*$), which is a necessary condition for optimal budget allocation across multiple placements.

Similar arguments apply to the scenario of optimizing a group of ads with a global budget. Under the optimal bidding strategy in our framework, budgets are allocated to each ad in an optimal manner to generate the most results. In fact, one major advantage of this framework is its composability. Automated budget allocation across placements, ad units as well as features described in Section \ref{other-types-of-constraint} are all elegantly handled via a unified bidding strategy.
\section{Practical Considerations}\label{practical-considerations}
In this section we aim to bridge the gaps between theory and practical implementation.  Important topics such as starting bid, forecasting error and various issues around the online updates are discussed. In particular, Section \ref{initialization-of-lambda} gives explicit formulas for the optimal starting bid based on statistics from the targeted audience. Section \ref{implementation-of-online-methods} describes several variants of the batch online gradient descent formula and applicability of them in different scenarios. 
\subsection{Initialization of $\lambda$}\label{initialization-of-lambda}
Section \ref{online-methods-in-the-dual} provides incremental update rules for the multiplier $\lambda$. The cold start problem, namely the starting value of the multiplier, is also very important as it impacts the time it takes to converge. On the macroscopic level it also influences the price stability of the entire marketplace. Given a set of targeted users, we can learn from these users' past auction logs to initialize $\lambda$.  In the following theorem we provide explicit formulas for second price auctions.
\begin{theorem}\label{thm:cold_start_explicit_solution}
Suppose the competitors' bids follow a log-normal distribution with parameters $\mu$ and $\sigma$. Additionally, suppose the value of the ad (independently) follows a log-normal distribution with parameters $\mu'$ and $\sigma'$. Given a budget of $B$ and a total opportunity forecast of $T$, the $\lambda^*$ for second price auctions is the solution to
\[
e^{\mu + \frac{\sigma^2}{2}}\Phi\left(\frac{\mu' - \mu - \ln \lambda^* - \sigma^2}{\sqrt{(\sigma')^2 + \sigma^2}}\right) = \frac{B}{T},
\]
where $\Phi$ is the CDF of a standard normal distribution $\mathcal N(0, 1)$.
\end{theorem}
In practice $\mu$ and $\sigma$ are derived using auction logs from targeted users. The $\mu'$ and $\sigma'$ come from a combination of targeted users' value distribution and ad-specific features. In the multiple placement scenario, the estimation of these quantities can be done at each placement level. The solution $\lambda^*$ will then be derived based on the traffic forecast $T_k$ from each placement $k$ and a global budget using Theorem \ref{thm:cold_start_explicit_solution}.
\subsection{Implementation of Online Methods}\label{implementation-of-online-methods}
Section \ref{online-methods-in-the-dual} provides various online update methods that converge to $\lambda^*$. In practice we employ a batch version of the dual online mirror descent algorithm, where the multiplier $\lambda_t$ is updated every time interval $dt$. In the following we focus on the additive formula in Theorem \ref{thm:dogd} though the arguments apply to the multiplicative update as well. Let $N_{dt}$ be the number of impression opportunities in time $dt$,  then the (batch) online update becomes
\begin{equation}\label{eq:batch_update}
\lambda_{t+N_{dt}} = \lambda_t - \epsilon\left(\frac{B}{T}\cdot N_{dt} - S_{dt}\right),
\end{equation}
where $S_{dt}$ is the sum of expected spend during the past update period under $\lambda_t$.  This can simply be set as the observed spend in $dt$. However if the charge events are sparse, some estimator based on the observed spend is needed to reduce the variance. \eqref{eq:batch_update} is equivalent to
\[
\lambda_{t+N_{dt}} = \lambda_t - \epsilon_{dt}\left(1 - R_{dt}\right),
\]
where
\[
\epsilon_{dt} := \epsilon\frac{B N_{dt}}{T}
\]
is the step size, and
\begin{equation}\label{eq:R_dt}
R_{dt} := \frac{S_{dt}/N_{dt}}{B/T}
\end{equation}
is the ratio of average cost per opportunity and average budget per opportunity. 

The choice of $dt$ is important as it controls the tradeoff between frequency of update and variance in $S_{dt}$.  An alternative is to keep track of the observed impression opportunities and to trigger an update after a fixed amount of observations. 
\subsubsection{Normalization}
The scale of $\lambda$ various a lot across advertisers, due to the diversity in budget and targeting setups. We therefore choose a normalizing constant $\lambda'$ for each advertiser and define $\tilde \lambda_t = \lambda_t/\lambda'$ to make the update dimensionless:
\[
\tilde \lambda_{t+N_{dt}} = \tilde \lambda_t - \eta_{dt}\left(1 - R_{dt}\right),
\]
where 
\begin{equation}\label{eq:eta_dt}
\eta_{dt} := \epsilon\frac{B N_{dt}}{\lambda' T}
\end{equation}
becomes a dimensionless step size that is easier to tune in practice.  However, the choice of $\lambda'$ for each advertiser plays an important role in convergence if a global parameter $\xi := \epsilon B/\lambda'$ is selected such that $\eta_{dt} = \xi N_{dt} / T$. To see that, note that the inverse of $\lambda_t$ will be used to compute the bids in auctions, and
\[
\frac{1}{\lambda_{t+N_{dt}}} = \frac{1}{\lambda'(\tilde \lambda_t - \eta_{dt}\left(1 - R_{dt}\right))} = \frac{1/\lambda_t}{1 - (\xi N_{dt} / T)\left(1 - R_{dt}\right)(\lambda' / \lambda_t)}.
\] 
Therefore a $\lambda' < \lambda^*$ would slow down the convergence whereas a $\lambda' > \lambda^*$ tends to create oscillatory behavior asymptotically due to magnification of the noise in $1 - R_{dt}$. In practice we find setting the normalization factor as the initialization provided in Section \ref{initialization-of-lambda} works well. 
\subsubsection{Total Forecast vs. Relative Forecast}
Sometimes it is easier to forecast the relative traffic pattern over time rather than the absolute number of opportunities $T$.  In that case one can use $\hat{r}_{dt}$, the forecasted proportion of traffic in time $dt$, to replace $N_{dt}$ and $T$ in the formula of $R_{dt}$ in \eqref{eq:R_dt} and $\eta_{dt}$ in \eqref{eq:eta_dt}. The formula then becomes $R_{dt} = S_{dt}/(B\hat{r}_{dt})$ and step size $\eta_{dt} = \epsilon B\hat{r}_{dt}/\lambda'$. The estimator $\hat{r}_{dt}$ can also be dynamically adjusted based on past observed traffic $N_{dt}$.
\subsubsection{Model Predictive Control}
One can reset the average budget per opportunity target $B/T$ at each update. A model predictive control (MPC) version of the update would set $R_{dt}$ as
\[
\frac{S_{dt}/N_{dt}}{(B-B_{t+N_{dt}})/(T-t-N_{dt})}.
\]
This version of the update would encourage budget exhaustion near the end of the ad's lifetime.
\section{Experimental Results}
Bidding agent based on this framework was implemented on the LinkedIn ad marketplace and is compared with a prior feedback-control agent which adjusts the bid to track the forecast traffic curve. For unbiased evaluation the two agents are compared using the budget-split experimentation platform \cite{L2020}. We observe statistically significant increase of 8.25\% in advertiser ROI (with neutral platform revenue) using the new method.
\section{Acknowledgement}
We would like to thank Yi Zhang and Onkar Dalal for their support. We appreciate Wen Pu, Qian Yao and Ricardo Salmon for helpful discussions.

\bibliographystyle{ACM-Reference-Format}
\bibliography{paper}

\appendix
\section{Proof of Theorem \ref{thm:dogd}}
This is standard result in online optimization \cite{H2019}. Ignoring the constant terms, we have
\[
\lambda_{t+1} = \text{argmin}_{\lambda \ge 0}\sum_{s=1}^t L'_s(\lambda_s)\cdot \lambda + R(\lambda).
\]
First order optimality condition gives $R'(\lambda_{t+1}) = -\sum_{s=1}^t L'_s(\lambda_s)$. In case $R(\lambda) = \lambda^2/(2\epsilon) \Rightarrow R'(\lambda) = \lambda/\epsilon$, then
\[
\lambda_{t+1} = -\epsilon\left(\sum_{s=1}^{t-1} L'_s(\lambda_s) + L'_t(\lambda_t)\right) = -\epsilon\left(\frac{\lambda_t}{-\epsilon} + L_t'(\lambda_t)\right) = \lambda_t -\epsilon L'_t(\lambda_t).
\]
Similarly,  $R(\lambda) = \lambda\ln\lambda / \epsilon \Rightarrow R'(\lambda) = (\ln\lambda + 1)/\epsilon$, then
\[
\lambda_{t+1} = e^{-1 - \epsilon\sum_{s=1}^{t-1}L'_s(\lambda_s) -\epsilon L'_t(\lambda_t)} = \lambda_te^{-\epsilon L'_t(\lambda_t)}.
\]
\section{Proof of Theorem \ref{thm:cold_start_explicit_solution}}
The following lemma is useful in the proof. Let $\phi$ and $\Phi$ be the PDF and CDF of a standard normal distribution $\mathcal N(0, 1)$.
\begin{lemma}\label{cor:Expectation_CDF}
Suppose $X \sim \mathcal N(0, 1)$, then $E[\Phi(aX+b)] = \Phi(\frac{b}{\sqrt{1+a^2}})$.
\end{lemma}
\begin{proof}
\[
\begin{aligned}
E[\Phi(aX+b)] &= \int_{-\infty}^{\infty}\int_{-\infty}^{ax+b}\phi(y)dy\phi(x)dx \\
&= \int_{-\infty}^{\infty}\int_{-\infty}^{\infty}\indicator{y \le ax+b}\cdot \phi(y)\phi(x)dydx.
\end{aligned}
\]
This is simply the probability of $Y \le aX+b$ given that $X, Y$ are i.i.d standard normally distributed variables. Since $Y - aX \sim \mathcal{N}(0, 1+a^2)$, this completes the proof.
\end{proof}
Now we proceed with the proof of the theorem.
\begin{proof}
Denote $p_{\mu, \sigma}(x)$ as the probability density function for $Lognormal(\mu, \sigma)$.  The expected cost per opportunity under a given $\lambda$ is
\[
S(\lambda) := \int_{0}^{\infty}\left(\int_{0}^{v/\lambda} zp_{\mu, \sigma}(z)dz\right)p_{\mu', \sigma'}(v)dv,
\]
where at each given value $v$ we would bid $v/\lambda$ (see Remark \ref{rm:explicit_formula_for_1st_and_2nd}).  First note that the inner integral $\int_{0}^{v/\lambda} zp_{\mu, \sigma}(z)dz$ is the partial expectation of a log-normal random variable.  By standard results
\[
\int_{0}^{v/\lambda} zp_{\mu, \sigma}(z)dz = e^{\mu + \frac{\sigma^2}{2}}\Phi\left(\frac{\ln v - \ln \lambda - \mu - \sigma^2}{\sigma}\right).
\]
That implies $S(\lambda) = e^{\mu + \frac{\sigma^2}{2}}E\left[\Phi\left(\frac{\ln V - \ln \lambda - \mu - \sigma^2}{\sigma}\right)\right]$, where $V$ is r.v. with $Lognormal(\mu', \sigma')$. Now with a change of variable $X := \frac{\ln V - \mu'}{\sigma'} \sim \mathcal{N}(0, 1)$, we have
\[
S(\lambda) = e^{\mu + \frac{\sigma^2}{2}}E\left[\Phi\left(\frac{\sigma' X + \mu' - \mu - \ln \lambda - \sigma^2}{\sigma}\right)\right].
\]
Now apply the results from Lemma \ref{cor:Expectation_CDF} ,
\[
S(\lambda) = e^{\mu + \frac{\sigma^2}{2}}\Phi\left(\frac{\mu' - \mu - \ln \lambda - \sigma^2}{\sqrt{(\sigma')^2 + \sigma^2}}\right).
\]
\end{proof}
\end{document}